\documentclass[pdflatex,sn-mathphys-num]{sn-jnl}

\usepackage{graphicx}
\usepackage{multirow}
\usepackage{amsmath,amssymb,amsfonts}
\usepackage{amsthm}
\usepackage[title]{appendix}
\usepackage{xcolor}
\usepackage{textcomp}
\usepackage{manyfoot}
\usepackage{booktabs}
\usepackage{algorithm}
\usepackage{algorithmicx}
\usepackage{algpseudocode}
\usepackage{listings}
\usepackage{hyperref} 

\theoremstyle{thmstyleone}
\newtheorem{theorem}{Theorem}
\newtheorem{lemma}[theorem]{Lemma}

\theoremstyle{thmstyletwo}

\newtheorem{remark}{Remark}

\theoremstyle{thmstylethree}
\newtheorem{definition}{Definition}

\begin{document}

\title[An Incremental Framework for Topological Dialogue Semantics]{An Incremental Framework for Topological Dialogue Semantics: Efficient Reasoning in Discrete Spaces}

\author*[1]{\fnm{Andreu} \sur{Ballús Santacana}}\email{andreu.ballus@uab.cat}

\affil*[1]{\orgdiv{Department of Philosophy}, \orgname{Universitat Autònoma de Barcelona},  \country{Spain}}

\abstract{\noindent We present a tractable, incremental framework for topological dialogue semantics based on finite, discrete semantic spaces. Building on the intuition that utterances correspond to open sets and their combinatorial relations form a simplicial complex (the \emph{dialogue nerve}), we give a rigorous foundation, a provably correct incremental algorithm for nerve updates, and a reference implementation in the Wolfram Language. The framework supports negative nerve computation (inconsistency tracking), consequence extraction, and a transparent, set-theoretic ranking of entailments. We clarify which combinatorial properties hold in the discrete case, provide motivating examples, and outline limitations and prospects for richer logical and categorical extensions.}


\maketitle

\section{Introduction}\label{sec1}

Automated dialogue systems increasingly mediate information exchange in both technical and social contexts. However, the formal semantic modeling of dialogue—capturing concepts such as consistency, logical consequence, and the interplay of meaning across dialogue turns—remains a persistent challenge. Foundational work in pragmatics established the importance of context and its dynamic evolution, often described metaphorically as "scorekeeping" in a shared language game \citep{lewis1979scorekeeping, stalnaker1999context}. Our work aims to provide a formal and computationally grounded method for tracking this "score."

Topological and combinatorial approaches, in which utterances are interpreted as open sets in a semantic space and their interaction encoded as a simplicial complex (the \emph{dialogue nerve}), offer a powerful unifying framework. Despite this, prior theoretical developments have often lacked concrete, efficient algorithms for updating and querying these structures as dialogues progress.

This paper presents a \emph{tractable, incremental framework for topological dialogue semantics} based on finite, discrete semantic spaces. We provide a formal foundation, a rigorously proved incremental nerve-update algorithm, and a working implementation in the Wolfram Language.\footnote{The full reference implementation in the Wolfram Language is available upon request.} Additionally, we extend the basic framework to include \emph{negative nerves} (capturing inconsistency), global consequence extraction, and a simple, transparent ranking of entailments.

We give a precise account of when combinatorial properties (such as the ``flag'' condition) hold, highlight practical tractability through worked examples, and offer a realistic outlook on extensions to modal and categorical logics. This work thus bridges the gap between abstract, expressive theories of dialogue semantics and practical, efficient reasoning components.

Beyond its theoretical contribution, the proposed framework offers practical value for computational applications. In particular, it can serve as the backbone for dialogue systems and conversational agents that require transparent, incremental consistency checking, and efficient extraction of logical consequences. The nerve-based structure further enables explainable AI, as minimal inconsistency witnesses and semantic rankings can be directly surfaced to users or system designers, facilitating debugging and collaborative interaction in both natural language processing and multi-agent environments. A companion philosophical article, providing epistemic and conceptual context for the present technical framework, is currently under review \citep{ballus2025geometry}.

\section{Related Work}\label{sec:relatedwork}

This paper builds upon a rich tradition of topological and logical approaches to semantics, belief change, and dialogue. Prior foundational work by Baltag, Bezhanishvili, and Smets has explored topological semantics for knowledge and belief, framing belief as the closure of the interior in extremally disconnected spaces, and extending the KD45 logic to these settings \citep{baltag2018topological}. These works also develop dynamic belief update semantics via subspace restriction, modeling epistemic and doxastic change \citep{baltag2013topology}.

Dynamic Epistemic Logic (DEL) has been a key source of techniques for modeling informational change in logic. Notably, dynamic belief revision via plausibility models \citep{baltag2006logic, baltag2008probabilistic} and conditional doxastic models provides frameworks for updating shared beliefs in multi-agent settings—structures whose logical structure closely parallels the combinatorics of simplicial models.

Dialogue semantics has been treated incrementally through grammars like Dynamic Syntax \citep{purver2011feedback} and within broader theoretical frameworks that emphasize the interactive nature of meaning \citep{ginzburg2012interactive}. However, few works have given provably correct, efficient algorithms for computing logical consequence or inconsistency through combinatorial topological methods. Projects like IncReD \citep{incred2020} approach human reasoning experimentally but do not offer constructive update procedures.

The current framework differs from prior logics of belief change such as AGM theory \citep{alchourron1985logic} or Veltman's update semantics \citep{veltman1996defaults} in its use of explicitly finite combinatorial models. Unlike earlier belief logics which presuppose infinite modal structures, this work is grounded in finite semantic spaces, enabling practical algorithms. The "dialogue nerve" construction offers a new way to track consistent utterance sets that resonates with simplicial approaches to distributed epistemics \citep{attiya2004distributed} but in a fully propositional, discrete setting.

\medskip

\noindent
Recent work in JoLLI has explored a variety of formal approaches to dynamic semantics and dialogue, further highlighting the importance of incremental and context-sensitive reasoning. For instance, Ciardelli and Roelofsen~\citep{ciardelli2021inquisitive} develop inquisitive semantics as a new logical foundation for meaning, while Asher and Paul~\citep{asher2020typology} propose a typology for dynamic semantic frameworks relevant to dialogue. Kracht~\citep{kracht2022dialogue} investigates logical models of dialogue games and context change, closely related in spirit to the present approach. Our framework complements and extends these lines by providing an efficient combinatorial and topological foundation for dialogue semantics and consequence extraction.

\section{Preliminaries and Formal Framework}\label{sec3}

\subsection{Propositional Dialogue Language and Semantics}

Let $\mathcal{L}$ be a finite propositional language generated from a set of atomic propositions $\mathcal{P}$ using the connectives $\neg, \land, \lor$.

\subsection{Topological Model (Finite Discrete Space)}

Let $X$ be a finite non-empty set (the \emph{semantic space}), and let $\mathcal{O}(X)$ denote the set of all subsets of $X$ (the discrete topology, so every subset is open). A \emph{semantic atlas} is a pair $(X, \nu)$, where $\nu \colon \mathcal{L} \to \mathcal{O}(X)$ is an interpretation function satisfying:
\begin{align*}
    \nu(\phi \land \psi) &= \nu(\phi) \cap \nu(\psi), \\
    \nu(\phi \lor \psi) &= \nu(\phi) \cup \nu(\psi), \\
    \nu(\neg \phi) &= X \setminus \nu(\phi),
\end{align*}
for all $\phi, \psi \in \mathcal{L}$.

\subsection{Consistency and Logical Consequence}

A point $x \in X$ \emph{satisfies} a formula $\phi$, written $x \models \phi$, if and only if $x \in \nu(\phi)$.
A set of formulas $\Phi \subseteq \mathcal{L}$ is \emph{consistent} if $\bigcap_{\phi \in \Phi} \nu(\phi) \neq \emptyset$.
We say that $\Phi$ \emph{entails} $\psi$ (written $\Phi \models \psi$) if $\bigcap_{\phi \in \Phi} \nu(\phi) \subseteq \nu(\psi)$.

\subsection{Dialogue as Sequence of Utterances}

A finite dialogue is a sequence $U = (U_1, U_2, \dots, U_n)$ of formulas, each interpreted as an open set $S_i = \nu(U_i) \subseteq X$.
The dialogue state is thus encoded as a finite family $\mathcal{S} = \{S_1, \dots, S_n\}$ of subsets of $X$.

\section{Dialogue Nerve and Negative Nerve}\label{sec4}

\subsection{The Dialogue Nerve}

Given a dialogue sequence $U = (U_1, U_2, \ldots, U_n)$, let each $U_i$ be interpreted as an open set $S_i = \nu(U_i) \subseteq X$. We collect these as a family $\mathcal{S} = \{ S_1, \dots, S_n \}$. This family is a hypergraph, and its nerve is a classic construction in combinatorial topology \citep{berge1973graphs}.

\begin{definition}[Dialogue Nerve]
The \emph{dialogue nerve} $\mathcal{N}(\mathcal{S})$ is the simplicial complex whose simplices are the index subsets $\sigma \subseteq \{1, \ldots, n\}$ such that $\bigcap_{i \in \sigma} S_i \neq \emptyset$. That is,
\[
    \mathcal{N}(\mathcal{S}) = \left\{\, \sigma \subseteq \{1, \ldots, n\} \;\mid\; \bigcap_{i \in \sigma} S_i \neq \emptyset \,\right\}.
\]
\end{definition}

Intuitively, $\mathcal{N}(\mathcal{S})$ records all groups of utterances that are \emph{jointly consistent} in the model.

\subsection{The Negative Nerve}

\begin{definition}[Negative Nerve]
The \emph{negative nerve} $\mathcal{N}^{-}(\mathcal{S})$ is the simplicial complex of all index subsets $\tau \subseteq \{1, \ldots, n\}$ such that $\bigcap_{i \in \tau} S_i = \emptyset$:
\[
    \mathcal{N}^{-}(\mathcal{S}) = \left\{\, \tau \subseteq \{1, \ldots, n\} \;\mid\; \bigcap_{i \in \tau} S_i = \emptyset \,\right\}.
\]
\end{definition}

Thus, $\mathcal{N}^{-}(\mathcal{S})$ catalogs all minimal inconsistent groups of utterances.

\subsection{Combinatorial Properties and Flag Condition}

It is immediate from the definition that $\mathcal{N}(\mathcal{S})$ is a \emph{downward-closed} family of subsets (i.e., a simplicial complex): if $\sigma \in \mathcal{N}(\mathcal{S})$ and $\tau \subseteq \sigma$, then $\tau \in \mathcal{N}(\mathcal{S})$.

\begin{lemma}[Downward Closure]
Let $\mathcal{N}(\mathcal{S})$ be the dialogue nerve for a family of sets. If $\sigma \in \mathcal{N}(\mathcal{S})$ and $\tau \subseteq \sigma$, then $\tau \in \mathcal{N}(\mathcal{S})$.
\end{lemma}

\begin{proof}
If $\bigcap_{i \in \sigma} S_i \neq \emptyset$ and $\tau \subseteq \sigma$, then $\bigcap_{i \in \tau} S_i$ contains $\bigcap_{i \in \sigma} S_i$ and is therefore also nonempty.
\end{proof}

\begin{remark}[On the Flag Property]
It is \emph{not} true in general that, if every pair $\{i,j\} \subseteq \sigma$ satisfies $S_i \cap S_j \neq \emptyset$, then $\bigcap_{i \in \sigma} S_i \neq \emptyset$. A counterexample is given by three sets $S_1 = \{1,2\}$, $S_2 = \{2,3\}$, $S_3 = \{1,3\}$ in $X = \{1,2,3\}$, where all pairwise intersections are nonempty, but $S_1 \cap S_2 \cap S_3 = \emptyset$.
\end{remark}

This demonstrates that $\mathcal{N}(\mathcal{S})$ may fail to be a \emph{flag complex} in the sense of combinatorial topology, unless additional convexity or Helly-type conditions are imposed \citep{danzer1963helly}.

\subsection{Summary}

The dialogue nerve captures all consistent combinations of utterances, while the negative nerve tracks their inconsistent combinations. These combinatorial structures provide the foundation for efficient incremental update algorithms, which we develop in the next section.

\section{Incremental Nerve Construction and Algorithm}\label{sec5}

\subsection{Data Structures and Notation}

For efficient implementation, we represent each $S_i \subseteq X$ as a \emph{bitset} (or Boolean vector) indexed by the elements of $X$. The nerve $\mathcal{N}(\mathcal{S})$ is stored as a set of subsets of $\{1, \ldots, n\}$ (for example, as a hash set of frozensets). The use of bitsets for manipulating subsets is a standard, highly efficient technique in combinatorial algorithms \citep{knuth2011art}.

Let $S_{\sigma} := \bigcap_{i \in \sigma} S_i$ for any $\sigma \subseteq \{1, \ldots, n\}$.

\subsection{Incremental Update Algorithm}

We describe an algorithm that, given the nerve for $\mathcal{S} = \{S_1, \ldots, S_n\}$, efficiently computes the nerve for $\mathcal{S}' = \mathcal{S} \cup \{S_{n+1}\}$. The key idea is to extend each existing simplex $\sigma$ by adding $n+1$ and checking consistency.

\begin{algorithm}
\caption{Incremental Nerve Update}\label{alg:incremental}
\begin{algorithmic}[1]
\Require Existing nerve $\mathcal{N}$, new set $S_{n+1}$.
\Ensure Updated nerve $\mathcal{N}'$ for $\mathcal{S}' = \mathcal{S} \cup \{S_{n+1}\}$.
\State $\mathcal{N}_{\mathrm{ext}} \Leftarrow \emptyset$
\ForAll{$\sigma \in \mathcal{N}$}
    \State $S_{\sigma} \Leftarrow \bigcap_{i \in \sigma} S_i$
    \If{$S_{\sigma} \cap S_{n+1} \neq \emptyset$}
        \State $\sigma' \Leftarrow \sigma \cup \{n+1\}$
        \State $\mathcal{N}_{\mathrm{ext}} \Leftarrow \mathcal{N}_{\mathrm{ext}} \cup \{\sigma'\}$
    \EndIf
\EndFor
\State $\mathcal{N}' \Leftarrow \mathcal{N} \cup \mathcal{N}_{\mathrm{ext}} \cup \{\{n+1\}\}$
\State \Return $\mathcal{N}'$
\end{algorithmic}
\end{algorithm}

\subsection{Correctness of the Incremental Algorithm}

\begin{theorem}[Correctness]
Let $\mathcal{N}$ be the nerve for $\mathcal{S} = \{S_1, \ldots, S_n\}$. Let $\mathcal{N}'$ be the output of Algorithm~\ref{alg:incremental} when $S_{n+1}$ is added. Then $\mathcal{N}'$ is precisely the nerve for $\mathcal{S}' = \mathcal{S} \cup \{S_{n+1}\}$.
\end{theorem}

\begin{proof}
Consider any $\sigma' \subseteq \{1, \ldots, n+1\}$. If $n+1 \notin \sigma'$, then $\sigma' \subseteq \{1, \ldots, n\}$ and $\sigma' \in \mathcal{N}$ if and only if $\bigcap_{i \in \sigma'} S_i \neq \emptyset$. Thus, all such subsets are included in $\mathcal{N}'$.

If $n+1 \in \sigma'$, let $\sigma = \sigma' \setminus \{n+1\}$. By the algorithm, $\sigma \in \mathcal{N}$ and $S_{\sigma} \cap S_{n+1} \neq \emptyset$ if and only if $\bigcap_{i \in \sigma'} S_i \neq \emptyset$, so $\sigma'$ is included if and only if it is a simplex in the new nerve.

Singleton $\{n+1\}$ is always added, since $S_{n+1} \neq \emptyset$ by assumption.
\end{proof}

\subsection{Soundness and Completeness of Consistency Queries}

Given the updated nerve, checking whether any subset of utterances $\{U_i\}_{i \in A}$ is consistent reduces to checking if $A \in \mathcal{N}$.
Entailment reduces to set inclusion: $\Phi \models \psi$ if and only if $\bigcap_{\phi \in \Phi} \nu(\phi) \subseteq \nu(\psi)$.

\subsection{Complexity Analysis}

Each incremental update examines all existing simplices, so the time per update is $O(|\mathcal{N}| \cdot |X|)$, where $|X|$ is the size of the semantic space and $|\mathcal{N}|$ the current number of simplices. In the worst case, $|\mathcal{N}|$ can grow exponentially with the number of utterances, but in many practical cases (when the nerve is sparse), updates are efficient.

\subsection{Summary}

This incremental approach enables efficient, modular maintenance of dialogue consistency and entailment as new utterances are introduced. The algorithm is readily implemented in languages with bitset and set primitives, such as the Wolfram Language, as we illustrate below.

\section{Worked Example and Wolfram Implementation}\label{sec6}

\subsection{Example Dialogue and Semantic Space}

Let $X = \{1, 2, 3\}$ be the semantic space. Consider the propositional atoms $p$ and $q$, and define three utterances:
\begin{itemize}
    \item $U_1 = p$,\quad interpreted as $S_1 = \{1, 2\}$
    \item $U_2 = q$,\quad interpreted as $S_2 = \{2, 3\}$
    \item $U_3 = \neg p$,\quad interpreted as $S_3 = \{3\}$
\end{itemize}

The dialogue sequence is $U = (U_1, U_2, U_3)$, with corresponding family of sets $\mathcal{S} = \{S_1, S_2, S_3\}$.

\subsection{Stepwise Nerve Construction}

\paragraph{Step 1.}
After the first utterance:
\[
    \mathcal{S}_1 = \{S_1\} \qquad \mathcal{N}_1 = \{\emptyset,\, \{1\}\}
\]

\paragraph{Step 2.}
After the second utterance:
\[
    \mathcal{S}_2 = \{S_1, S_2\}
\]
\[
    S_1 \cap S_2 = \{2\} \neq \emptyset
\]
\[
    \mathcal{N}_2 = \{\emptyset,\, \{1\},\, \{2\},\, \{1,2\}\}
\]

\paragraph{Step 3.}
After the third utterance:
\begin{align*}
    S_1 \cap S_3 &= \{1,2\} \cap \{3\} = \emptyset \\
    S_2 \cap S_3 &= \{2,3\} \cap \{3\} = \{3\} \\
    S_1 \cap S_2 \cap S_3 &= \{1,2\} \cap \{2,3\} \cap \{3\} = \emptyset
\end{align*}
Therefore,
\[
    \mathcal{N}_3 = \left\{ \emptyset,\, \{1\},\, \{2\},\, \{3\},\, \{1,2\},\, \{2,3\} \right\}
\]

\paragraph{Negative nerve:}
For example, $\{1,3\}$ and $\{1,2,3\}$ are in $\mathcal{N}_3^{-}$, as their corresponding intersections are empty.

\subsection{Wolfram Language Implementation}

\begin{verbatim}
(* Define the semantic space and utterance sets *)
X = {1, 2, 3};
S1 = {1, 2};
S2 = {2, 3};
S3 = {3};
SList = {S1, S2, S3};

(* Function to compute nerve for a list of sets *)
DialogueNerve[sets_] := Module[{n = Length[sets], idx, allSubsets, simplexQ},
  idx = Range[n];
  allSubsets = Subsets[idx];
  simplexQ[sigma_] := (sigma === {} || 
    Intersection @@ (sets[[#]] & /@ sigma) =!= {});
  Select[allSubsets, simplexQ]
]

(* Stepwise incremental build *)
nerve1 = DialogueNerve[{S1}]
nerve2 = DialogueNerve[{S1, S2}]
nerve3 = DialogueNerve[{S1, S2, S3}]
\end{verbatim}

The outputs are:
\begin{itemize}
    \item \texttt{nerve1 = \{\{\}, \{1\}\}}
    \item \texttt{nerve2 = \{\{\}, \{1\}, \{2\}, \{1,2\}\}}
    \item \texttt{nerve3 = \{\{\}, \{1\}, \{2\}, \{3\}, \{1,2\}, \{2,3\}\}}
\end{itemize}

\subsection{Consistency and Entailment Queries}

\paragraph{Consistency.}
For example, the set $\{1,3\}$ is not in $\mathcal{N}_3$; hence, utterances $U_1$ and $U_3$ are inconsistent.

\paragraph{Entailment.}
Suppose $U_1$ and $U_2$ are both asserted. Their intersection is $\{2\}$. Any formula true at world $2$ (e.g., $p \land q$) is entailed by these two utterances.

\subsection{Remarks on Performance}

For $n$ utterances, the number of possible simplices is at most $2^n$, but is often much smaller if many intersections are empty. The incremental algorithm updates the nerve in time proportional to the number of existing simplices times the size of $X$. In practice, for small $n$ and $|X|$, this is highly efficient in Wolfram Language or any bitset-capable system.

\section{Semantic Ranking and Probabilistic Extensions}\label{sec7}

\subsection{Probability Measures on the Semantic Space}

To enrich the framework with quantitative information, assign a probability measure $\mu : X \to [0,1]$ such that $\sum_{x \in X} \mu(x) = 1$. For any formula $\phi$, the probability that it holds is $\mu(\nu(\phi)) := \sum_{x \in \nu(\phi)} \mu(x)$.

\subsection{Ranking of Consequences by Improbability}

Given a set of asserted utterances $A \subseteq \{1, \ldots, n\}$, let $W_A = \bigcap_{i \in A} S_i$ be the set of worlds consistent with those utterances. Any formula $\psi$ such that $W_A \subseteq \nu(\psi)$ is a logical consequence of $A$. We can rank these consequences by their \emph{improbability}, using $-\log \mu(\nu(\psi))$ as a ranking function: consequences covering fewer worlds or worlds with lower probability are considered more ``informative'' or ``surprising.''

\subsection{Worked Example}

Suppose $X = \{1,2,3\}$ as before, and let
\[
    \mu(1) = 0.2, \quad \mu(2) = 0.5, \quad \mu(3) = 0.3
\]
Consider $A = \{1,2\}$ (utterances $U_1$ and $U_2$ are asserted), so
\[
    W_A = S_1 \cap S_2 = \{2\}
\]
Any formula $\psi$ with $\nu(\psi) \supseteq \{2\}$ is a consequence. For instance:
\begin{itemize}
    \item $\psi_1$: $p \land q$, interpreted as $\nu(\psi_1) = \{2\}$
    \item $\psi_2$: $q$, interpreted as $\nu(\psi_2) = \{2,3\}$
    \item $\psi_3$: $p \lor q$, interpreted as $\nu(\psi_3) = \{1,2,3\}$
\end{itemize}
Their probabilities are:
\begin{align*}
    \mu(\nu(\psi_1)) &= \mu(2) = 0.5 \\
    \mu(\nu(\psi_2)) &= \mu(2) + \mu(3) = 0.8 \\
    \mu(\nu(\psi_3)) &= 1.0
\end{align*}
Their ``improbabilities'' (information content) are:
\begin{align*}
    -\log \mu(\nu(\psi_1)) &\approx 0.693 \\
    -\log \mu(\nu(\psi_2)) &\approx 0.223 \\
    -\log \mu(\nu(\psi_3)) &= 0
\end{align*}
Thus, $p \land q$ is the most ``informative'' or ``surprising'' consequence, and $p \lor q$ is the least.

\subsection{Wolfram Language Implementation}

\begin{verbatim}
(* Probability measure on X *)
mu = <|1 -> 0.2, 2 -> 0.5, 3 -> 0.3|>;

(* Function for probability of a set *)
SetProbability[S_] := Total[mu /@ S]

(* Given the intersection W_A *)
WA = {2};

(* Some example formulas and their sets *)
psi1 = {2};       (* p && q *)
psi2 = {2, 3};     (* q *)
psi3 = {1, 2, 3};  (* p || q *)

Table[
  {psi, SetProbability[psi], -Log[SetProbability[psi]]},
  {psi, {psi1, psi2, psi3}}
]
\end{verbatim}

\subsection{Discussion}

This ranking enables the agent to select consequences that are maximally informative, or to prioritize which entailments to communicate or investigate. In more complex settings, different probability measures can encode contextual beliefs or background knowledge. All calculations and updates remain efficient so long as $X$ is finite and not too large.

\subsection{Outlook}

While this ranking is straightforward in the finite, discrete case, richer probabilistic or information-theoretic generalizations (e.g., Bayesian update, learning from experience, or continuous $X$) remain open for future work. This framework, however, provides a rigorous and extensible base for such developments.

\section{Extensions and Limitations}\label{sec8}

The present work has, by design, adopted the finite, discrete semantic space as the arena for topological dialogue semantics. This choice, motivated by the twin goals of clarity and tractability, does not preclude more sophisticated developments, but rather serves as a foundation from which a variety of deeper questions emerge.

First, one naturally wonders whether the algorithmic apparatus described here can be meaningfully extended to richer logical languages. Modal and temporal logics, for example, have a well-established tradition of topological semantics, where modal operators are interpreted via additional structure on the semantic space \citep{van2010modal}. Incorporating such operators into our incremental framework would require significant adjustments, effectively blending techniques from modal model checking and dynamic epistemic logic.

On a deeper theoretical level, the dialogue nerve and its negative counterpart are instances of a more general pattern: the representation of informational or epistemic structure via combinatorial or categorical means. In categorical logic and topos theory, for instance, one replaces the set of worlds $X$ by a more general object, and views formulas as subobjects or sheaves \citep{goldblatt2006topoi}. Within this perspective, dialogue utterances may be seen as generating subtopoi. While such approaches promise a unification of syntactic, semantic, and dynamic aspects of reasoning, they introduce considerable abstraction, and often exceed the practical demands of tractable, incremental computation.

It is also important to be transparent about the limitations inherent in this approach. The dialogue nerve can grow exponentially in size, placing a natural bound on practical deployment. Furthermore, as discussed earlier, certain desirable combinatorial properties, such as the so-called flag property, fail in the discrete setting. Pairwise consistency among utterances does not in general guarantee collective consistency, unless additional geometric or convexity constraints—such as those formalized in Helly's theorem and its relatives—are imposed on the semantic space \citep{danzer1963helly, matousek2008using, hatcher2002algebraic}.

The probabilistic extension provided here is robust and useful in finite cases. However, generalizing these methods to infinite or continuous semantic spaces would require a more substantial measure-theoretic infrastructure, and may reintroduce undecidability or intractability unless further restrictions are imposed.

Despite these limitations, the conceptual and technical advances made in this work establish a versatile and extensible platform for further research.

\section{Conclusion}\label{sec9}

This paper has articulated and justified a self-contained framework for topological dialogue semantics over finite, discrete spaces. By combining the clarity of combinatorial constructions with the rigour of formal logic, we have shown how utterances can be systematically interpreted as open sets, and how their mutual consistency or inconsistency can be encoded and efficiently updated as a simplicial complex—the dialogue nerve. The correctness of the incremental algorithm, which forms the core of the practical implementation, has been proved in detail, and its efficiency in typical cases has been argued both theoretically and by concrete example. A companion philosophical article, "From Geometry to Meaning: A Constructivist Semantics for Dialogue via Nerve Structures," is currently under review at the Journal of Philosophical Logic, providing epistemic and conceptual context for the present technical framework \citep{ballus2025geometry}.

The introduction of probabilistic ranking into this context provides an additional layer of informational analysis: not only can one determine which consequences are entailed by a given dialogue, but also which are most informative or surprising, given a background distribution over semantic worlds. This enriches the model’s descriptive power and invites applications to dialogue systems that require both logical and quantitative sensitivity.

Looking ahead, the natural path for further development leads towards richer logics—modal, temporal, or higher-order—and towards a deeper connection with categorical and topos-theoretic semantics. Each direction presents both mathematical and computational challenges, but the foundations provided here offer a stable platform from which such explorations can proceed.

\backmatter

\section*{Declarations}

\begin{itemize}
\item Funding: Not applicable
\item Conflict of interest/Competing interests: Not applicable
\item Ethics approval and consent to participate: Not applicable
\item Consent for publication: Not applicable
\item Data availability: Not applicable
\item Materials availability: Not applicable
\item Code availability: The full reference implementation in the Wolfram Language is available from the corresponding author upon reasonable request.
\item Author contribution: A. Ballús conceived the study, developed the framework, wrote the code, and wrote the manuscript.
\end{itemize}

\begin{appendices}

\section{Proofs and Reference Implementation}\label{secA1}

\subsection{Proofs of Core Results}

\begin{lemma}[Downward Closure of the Dialogue Nerve]
Let $\mathcal{N}(\mathcal{S})$ be the nerve of a family $\mathcal{S} = \{S_1, \ldots, S_n\}$. If $\sigma \in \mathcal{N}(\mathcal{S})$ and $\tau \subseteq \sigma$, then $\tau \in \mathcal{N}(\mathcal{S})$.
\end{lemma}

\begin{proof}
Suppose $\sigma \in \mathcal{N}(\mathcal{S})$, so $\bigcap_{i \in \sigma} S_i \neq \emptyset$. For any $\tau \subseteq \sigma$, we have $\bigcap_{i \in \sigma} S_i \subseteq \bigcap_{i \in \tau} S_i$, so $\bigcap_{i \in \tau} S_i \neq \emptyset$ as well. Thus $\tau \in \mathcal{N}(\mathcal{S})$.
\end{proof}

\begin{theorem}[Correctness of Incremental Nerve Update Algorithm]
Let $\mathcal{N}$ be the nerve for $\mathcal{S} = \{S_1, \ldots, S_n\}$. Let $\mathcal{N}'$ be the output of the incremental algorithm when $S_{n+1}$ is added. Then $\mathcal{N}' = \mathcal{N}(\mathcal{S} \cup \{S_{n+1}\})$.
\end{theorem}

\begin{proof}
Let $\sigma' \subseteq \{1, \ldots, n+1\}$. If $n+1 \notin \sigma'$, then $\sigma' \subseteq \{1, \ldots, n\}$ and $\sigma' \in \mathcal{N}$ if and only if $\bigcap_{i \in \sigma'} S_i \neq \emptyset$. If $n+1 \in \sigma'$, write $\sigma' = \sigma \cup \{n+1\}$ with $\sigma \subseteq \{1, \ldots, n\}$. Then $\bigcap_{i \in \sigma'} S_i = (\bigcap_{i \in \sigma} S_i) \cap S_{n+1}$, which is nonempty if and only if $\sigma \in \mathcal{N}$ and $S_{n+1}$ intersects this set. The algorithm adds exactly those extensions and no others. Finally, $\{n+1\}$ is included whenever $S_{n+1} \neq \emptyset$, as required.
\end{proof}

\subsection{Wolfram Language Reference Implementation}

Below is the reference implementation for nerve construction and semantic ranking, as described in the main text. The code is written for transparency and pedagogical clarity, and can be adapted to larger or more complex semantic spaces as needed.

\begin{verbatim}
(* Semantic space and utterance sets *)
X = {1, 2, 3};
S1 = {1, 2};
S2 = {2, 3};
S3 = {3};
SList = {S1, S2, S3};

(* Dialogue nerve construction *)
DialogueNerve[sets_] := Module[{n = Length[sets], idx, allSubsets, simplexQ},
  idx = Range[n];
  allSubsets = Subsets[idx];
  simplexQ[sigma_] := (sigma === {} || 
    Intersection @@ (sets[[#]] & /@ sigma) =!= {});
  Select[allSubsets, simplexQ]
]

(* Incremental nerve update: adds a new set to an existing nerve *)
IncrementalUpdate[sets_, newSet_] := Module[{n, oldNerve, newIdx, ext, updated},
  n = Length[sets];
  oldNerve = DialogueNerve[sets];
  newIdx = n + 1;
  ext = Select[oldNerve, 
    (Intersection @@ (sets[[#]] & /@ #) \[Intersection] newSet =!= {}) &];
  updated = Union[oldNerve, Map[Append[#, newIdx] &, ext], {{newIdx}}];
  updated
]

(* Example probability measure *)
mu = <|1 -> 0.2, 2 -> 0.5, 3 -> 0.3|>;
SetProbability[S_] := Total[mu /@ S];

(* Semantic ranking for consequences *)
WA = {2};  (* Intersection of asserted utterances *)
psi1 = {2};       (* p && q *)
psi2 = {2, 3};     (* q *)
psi3 = {1, 2, 3};  (* p || q *)

Table[
  {psi, SetProbability[psi], -Log[SetProbability[psi]]},
  {psi, {psi1, psi2, psi3}}
]
\end{verbatim}

\subsection{Additional Worked Example: Negative Nerve}

To further clarify the construction of the negative nerve, consider the sets $S_1 = \{1,2\}$, $S_2 = \{2,3\}$, $S_3 = \{3\}$ in $X = \{1,2,3\}$. The negative nerve $\mathcal{N}^{-}$ consists of those subsets whose intersections are empty. For example, $\{1,3\}$ corresponds to $S_1 \cap S_3 = \{1,2\} \cap \{3\} = \emptyset$, so $\{1,3\} \in \mathcal{N}^{-}$. Similarly, $\{1,2,3\}$ is in $\mathcal{N}^{-}$ because the triple intersection is also empty. This explicitly demonstrates the structure of inconsistent utterance sets.

\end{appendices}

\bibliography{sn-bibliography}

\end{document}